\documentclass[9pt,a4j]{article}
\usepackage{setspace}
\usepackage{amsmath,amsthm,bm,amssymb,extarrows,ascmac}
\usepackage{graphicx}
\usepackage{tikz}
\usetikzlibrary{intersections,shapes.arrows,calc}

\newtheorem{theo}{Theorem}[section]
\newtheorem{defi}[theo]{Definiton}
\newtheorem{lemm}[theo]{Lemma}%
{\tiny }%

\newcommand{\argmin}{\mathop{\rm argmin}\limits}

\title{\ Geometry of Arimoto algorithm}%
\author{Shoji Toyota}%
\date{}%

\begin{document}
    \maketitle
    
\begin{abstract}
In information theory, the channel capacity, which indicates how efficient a given channel is, plays an important role.  The best-used algorithm for evaluating the channel capacity is Arimoto algorithm \cite{Arimoto}. This paper aims to reveal an information geometric structure of Arimoto algorithm.   In the process of trying to reveal an information geometric structure of Arimoto algorithm, a new algorithm that monotonically increases the Kullback-Leibler divergence is proposed, which is named “the Backward em-algorithm.'' Since the Backward em-algorithm is available in many cases where we need to increase the Kullback-Leibler divergence, it has a rich potential for application to many problems of statistics and information theory.
\end{abstract} 

\tableofcontents

\newpage
    \section{Introduction}
    Since  C. E. Shannon proposed the notion of channel capacity \cite{Shannon}, it has played an important role in information theory. Given a channel ($\Omega_1, r(y|x), \Omega_2$), the channel capacity $C$ is defined as follows:
     \begin{eqnarray*}
    C:=\sup_{q(x)\in  S_1} I(q(x)\cdot r(y|x)),
    \end{eqnarray*}
    Here, $\Omega_1$ and $\Omega_2$ denote finite sets and $r(y|x)$ denotes a conditional probability on $\Omega_2$ for $x \in \Omega_1$. The symbol $I$ denotes the mutual information of $q(x)\cdot r(y|x)$ and $\mathcal{S}_1$ denotes the set of all probability distributions on $\Omega_1$.
       
Arimoto algorithm \cite{Arimoto} is known as the best-used algorithm for evaluating the channel capacity of a memoryless channel, where we update $q^{(t)}(x) \in \mathcal{S}_1$ in order that $I(q^{(t)}(x)\cdot r(y|x))$ increases. Although many people have proposed other algorithms (e.g., \cite{Yu}), they are essentially the same as Arimoto algorithm. It implies that Arimoto algorithm is not just an algorithm but has some generic structure. The purpose of this paper is to reveal a theoretical justification of Arimoto algorithm from the information geometric point of view. 

There exist papers whose purpose are similar to the present paper, for example \cite{Ikeda}, \cite{Kenji}, \cite{Matz} and \cite{Csisza'r}. But \cite{Ikeda} and \cite{Kenji} mention only the channel capacity but not Arimoto algorithm. Although \cite{Matz} refers to Arimoto algorithm, we think it does not sufficiently explain a theoretical justification of Arimoto algorithm from the information geometric point of view (see Section \ref{Inf.geom in 3} for more information). The paper \cite{Csisza'r} tries to interpret Arimoto algorithm by using the Kullback-Leibler divergence. But, to do so, \cite{Csisza'r} expands its domain outside of the probability simplex. Since Information Geometry is conventionally geometric structures on the probability simplex ($i.e.$, ``inside'' of the probability simplex), to reveal information geometric view of Arimoto algorithm, further studies are needed. Since our analysis is inside of the probability simplex, it can be said that we deal with more generic information geometric view than previous studies.

This paper is organized as follows. In Section \ref{terminologies of Information Geometry}, we summarize some terminologies and results of Information Geometry. In Section \ref{Section.Def}, we explain the channel capacity and Arimoto algorithm. Information geometric view of a channel capacity is
     investigated in Section \ref{Inf.geom in 3}. In Section \ref{Backward em-algorithm}, we propose an algorithm naturally induced from the information geometric view of a channel capacity addressed in Section \ref{Inf.geom in 3}, and prove that this algorithm corresponds to Arimoto algorithm. We conclude the paper with brief remarks in Section \ref{section.Concluding Remarks}. 

\section{\bf Information Geometry} \label{terminologies of Information Geometry}

In a narrow sense, Information Geometry on a finite model is a dually flat structure on the probability simplex. In this section, we summarize one of the dually flat structure used in the present paper. 

\begin{defi}

Let N be a $C^{\infty}$ manifold, $g$ be a Rieamannian metric on $N$ and $\nabla$, $\nabla^*$ be affine connections on $N$. We call the triple $( g, \nabla, \nabla^*)$ an dually structure on $N$ if
\begin{eqnarray*}
Xg(Y, Z) = g(\nabla_{X}Y, Z) + g(X, \nabla^*_{Y} Z) ~~~ (\forall X,Y,Z \in \chi(N))
\end{eqnarray*}
holds. Here, $\chi(N)$ denotes the set of all vector fields on $N$. Especially, if  $\nabla$ and $\nabla^*$ are flat, we call the triple $( g, \nabla, \nabla^*)$ a dually flat structure on N.
\end{defi}
Let $\Omega$ be a finite set. We can regard the probability simplex
\begin{eqnarray*}
\mathcal{S}  :=\{p:\Omega \rightarrow \mathbb{R}_{++} ;\sum_{x \in \Omega} p(x)=1\}, 
\end{eqnarray*}
as an  $(|\Omega|-1)$-dimensional submanifold of $\mathbb{R}^{n}$.
The Fisher metric $g$ and the m-connection $\nabla^{(m)}$ and the e-connection $\nabla^{(e)}$ on $\mathcal{S}$ are defined as follows:
\begin{eqnarray*}
g_p ( X, Y) := \sum_{\omega \in \Omega} p(\omega) ( X \log p(\omega) ) ( Y \log p(\omega)),\\
g_p ( \nabla_{X}^{(m)} Y, Z):=  g_p ( \overline{\nabla}_{X} Y, Z) + \frac{\alpha}{2} S_p (X, Y, Z),\\
g_p ( \nabla_{X}^{(e)} Y, Z):=  g_p ( \overline{\nabla}_{X} Y, Z) - \frac{\alpha}{2} S_p (X, Y, Z).
\end{eqnarray*}
Here, $ \overline{\nabla}$ denotes the Levi-Civita connection of $g$ and $S$ denotes the $(0,3)$-tensor on $S$ defined by
\begin{eqnarray*}
S_p ( X, Y, Z) := \sum_{\omega \in \Omega} p(\omega) ( X \log p(\omega) ) ( Y \log p(\omega)) ( Z \log p(\omega)).
\end{eqnarray*}
Note that the triple $(g, \nabla^{(m)}, \nabla^{(e)})$ is a dually flat structure on $\mathcal{S}$ \cite[p.35, p36, Theorem 3.1]{S. Amari and H. Nagaoka}. $\nabla^{(m)}$ and $\nabla^{(e)}$ have the global affine coordinate systems $(\eta_i)^{|\Omega|-1}_{i=1}$ and $(\theta_j)^{|\Omega|-1}_{j=1}$ defined as follows:
$$\eta_i := p(i),$$
$$\theta_j := \log{\frac{p(j)}{p(|\Omega|)}}.$$
$\nabla^{(m)}$-geodesics and $\nabla^{(e)}$-geodesics have the following interesting property.
\begin{theo}\cite[Theorem 3.8]{S. Amari and H. Nagaoka} 
Let $p, q$ and $r$ be elements of $\mathcal{S}$. Asssume that the $\nabla^{(m)}$-geodesic connecting $p$ and $q$ and $\nabla^{(e)}$-geodesic connecting $q$ and $r$ are orthogonal at q. Then,
\begin{equation}
D(p||q) + D(q||r) = D(p||r)
\end{equation}
holds. Here, $D( ~|| ~)$ denotes the Kullback Leibler divergence defined as follows:
$$ D(p_1 || p_2 ) := \sum_{ x \in \Omega} p_1 (x) \log { \frac{p_1(x)}{p_2(x)}} \qquad (p_1, p_2 \in \mathcal{S}).$$
\end{theo}

The above theorem is called the generalized Pythagorean theorem.

Next, we define $\nabla^{(m)}$-projections and $\nabla^{(e)}$-projections.

\begin{defi}
Let $K$ be a submanifold of $\mathcal{S}$ and $p \in \mathcal{S}$. We call $\hat{p}$ a $\nabla^{(m)}$- resp. $\nabla^{(e)}$- projection of $p$ onto $K$ if 
the $\nabla^{(m)}$- resp. $\nabla^{(e)}$-geodesic connecting $p$ and  $\hat{p}$ are ``orthogonal'' to $K$ (with respect to the Fisher metric $g$) at $\hat{p}$.
\end{defi}

In general, a $\nabla^{(m)}$-projection nor a $\nabla^{(e)}$-projection is unique. But, if $K$ has the following property, the projection becomes unique.
\begin{defi}
Let $K$ be a submanifold of $\mathcal{S}$. We say that $K$ is $\nabla^{(m)}$-autoparallel if, for any $X, Y \in \chi(K)$, $\nabla^{(m)}_X Y |_{p} \in T_p (K)$. Similarly, $K$ is said to be  $\nabla^{(e)}$-autoparallel if, for any $X, Y \in \chi(K)$, $\nabla^{(e)}_X Y |_{p} \in T_p (K)$. Here, $T_p (K)$ denotes the tangent space of $K$ at $p$ embedded into the tangent space $T_p (\mathcal{S})$.
\end{defi}

\begin{theo}\cite[Theorem 3.9]{S. Amari and H. Nagaoka}
Let  $M$ and $E$ be  $\nabla^{(m)}$-autoparallel and $\nabla^{(e)}$-autoparallel submanifolds in $\mathcal{S}$ respectively. Let $p \in \mathcal{S}$. Then, a necessary and sufficient condition for $\hat{p}$ to be a $\nabla^{(m)}$-projection of $p$ onto $E$ is that $\hat{p}$ satisfies
\begin{equation*} 
\hat{p } = \argmin_{ p \in E} D(\hat{p} || p).
\end{equation*}
And the $\nabla^{(m)}$-projection onto $E$ is unique if it exists.

Similarly, a necessary and sufficient condition for $\hat{p}$ to be a $\nabla^{(e)}$-projection of $p$ onto $M$ is that $\hat{p}$ satisfies
\begin{equation*} 
\hat{p} = \argmin_{ p \in M} D(p  || \hat{p} ).
\end{equation*}
And the $\nabla^{(e)}$-projection onto $M$ is unique if it exists.

\end{theo}

We often need to investigate whether or not a submanifold $M$ of $\mathcal{S}$ is  $\nabla^{(m)}$ or $\nabla^{(e)}$-autoparallel. The following theorem gives a sufficient condition for $M$ to be $\nabla^{(m)}$- and $\nabla^{(e)}$-autoparallel.

\begin{theo}\label{Theo.sufficient condition to autoparallel}
Assume that, for any $p_1, p_2 \in M$ and $ t \in (0,1)$, the element 
\begin{equation}\label{eq. autoparallel w.r.t m}
 p_3 := t p_1 + (1-t) p_2 \in \mathcal{S}
\end{equation}
belong to $M$.

Then, $M$ is $\nabla^{(m)}$- autoparallel.

Let $E$ be a submanifold of $\mathcal{S}$. Assume that, for any $p_1,p_2 \in E$ and $t \in (0,1)$, the element $p_3$ for which
$$ \log p_3 = t\log p_1 + (1-t) \log p_2 + A$$
belong to $E$. Then $E$ is $\nabla^{(e)}$-autoparallel.  Here, the constant $A$, which is independent of $\omega \in \Omega$, is defined by 
\begin{equation}
A:= \log  \left\{ \left\{ \sum_{\omega \in \Omega} \exp \left\{ t \log p_1 (\omega) + (1-t) \log p_2 (\omega) \right\} \right\}^{-1} \right\}.
\end{equation}
\end{theo}

To prove Theorem \ref{Theo.sufficient condition to autoparallel}, we need the following lemmma.

\begin{lemm}\cite[Theorem 3.7.3]{Fujiwara}
Let $m, n \in \mathbb{N}$ with $0 < m \leq n$. Let $N$ be an n-dimensional flat manifold with respect to the affine connection $\nabla$ and there exists a global affine coordinate system $(x^i)_{1 \leq i \leq n}$ of $N$. A necessary and sufficient condition for an m-dimensional submanifold $M$  of $N$ to be autoparallel is that there exists a local coordinate system $(\xi^{a})_{1 \leq a \leq m}$, an ($m \times n$)-matrix  $ A$ such that $rankA = m$ and $b \in \mathbb{R}^n$ which satisfy
$$\left(
\begin{array}{ccccc}
x^{1} \\
\vdots \\
x^{m} \\
x^{m+1}\\
\vdots \\
x^{n} 
\end{array}
\right) = 
A
\left(
\begin{array}{ccccc}
\xi^{1} \\
\vdots \\
\xi^{m} 
\end{array}\right) 
+ b.
$$
\end{lemm}\label{Lemm.aaaaaaa}
\begin{proof}[\bf {Proof of Theorem \ref{Theo.sufficient condition to autoparallel}}]
Assume that $M$ satisfies the equation (\ref{eq. autoparallel w.r.t m}). Then $M$ is convex with respect to the $\nabla^{(m)}$-affine coordinate system $(\eta_i)^{|\Omega|-1}_{i=1}$. Fix $p \in M$. Take $p_1(\neq p) \in M$. If $m \geq 2$, we can take $p_2 \in M$ such that $p_1 -p$ and $p_2 - p$ are linearly independent. Repeating this, we can take $p_1, ..., p_m \in M$ such that $p_1 - p ,...,p_m -p $ are linearly independent. Define the ``hyperplane'' (with respect to $(\eta_i)^{|\Omega|-1}_{i=1}$) $T(p,p_1,...,p_m)$ by 
$$T(p,p_1,...,p_m) := \left\{ \hat{p} \in \mathcal{S} \left|  \hat{p} = p  +  \left[ p_1-p,...,p_m-p \right]
\left(
\begin{array}{ccccc}
\xi^{1} \\
\vdots \\
\xi^{n} 
\end{array}\right) ,
~(\xi^{a} )_{a=1}^{m} \in \mathbb{R}^m \right. \right\}.$$
Since $p_1 - p ,...,p_m - p $ are linearly independent, we can see that $rank\left[ p_1 -p,...,p_m - p\right] = m$. Noting that $M$ is a submanifold of $T(p,p_1,...,p_m)$  and $dim(M) = dim(T(p,p_1,...,p_m))$,  we can see that there exists a local coordinate system $(\xi^{a})_{1 \leq a \leq m}$ of $M$ such that
$$\left(
\begin{array}{ccccc}
\eta^{1} \\
\vdots \\
\eta^{m} \\
\eta^{m+1}\\
\vdots \\
\eta^{n} 
\end{array}
\right) = 
\left[ p_1 -p,...,p_m -p\right]
\left(
\begin{array}{ccccc}
\xi^{1} \\
\vdots \\
\xi^{m} 
\end{array}\right) 
+ p
$$
holds. From Lemma \ref{Lemm.aaaaaaa}, we can see that $M$ is $\nabla^{(m)}$-autoparallel. The proof of the latter half is same as the above proof.
\end{proof}

\section{\bf Channel capacity and Arimoto algorithm} \label{Section.Def}

In this paper, let $\Omega_i$ ($i=1,2$) be finite sets, $\mathcal{S}_i$ be the sets of all probability distributions on $\Omega_i$. Namely, 
        \begin{eqnarray*}
    \mathcal{S}_i  :=\{p:\Omega_i \rightarrow \mathbb{R}_{++} ;\sum_{x \in \Omega_i} p(x)=1\}~~~~(i=1,2), 
    \end{eqnarray*}
    where $\mathbb{R}_{++}:=\{x\in \mathbb{R};x>0\}$. 
    Similarly, let $\mathcal{S}_3$ be the set consisting of all probability distributions on  $\Omega_1 \times \Omega_2$.
    
A memoryless channel is expressed by a system where, for an input symbol $x\in \Omega_1$, an output symbol $y \in \Omega_2$ is determined at random.
    \begin{defi}
    A channel is defined by a triple $(\Omega_1,r(y|x),\Omega_2)$ of finite sets $\Omega_1,\Omega_2$ and a map $\Omega_1 \owns x\mapsto r(\cdot |x) \in \mathcal{S}_2$  .\end{defi}
    \begin{defi}
    We call the map $I:\mathcal{S}_3 \rightarrow \mathbb{R}$ defined by 
    \begin{eqnarray}
    I(p(x,y)):=D(p(x,y)||q(x)\cdot r(y))\label{Mutual}
    \end{eqnarray}
    the mutual information. In the equation $(\ref{Mutual})$, $q(x)$ and $r(y)$ mean the marginal distributions of $p(x,y)$ on $\Omega_1$ and  $\Omega_2$ respectively.
     \end{defi}
     
     \begin{defi}\label{capacity} Given a channel $(\Omega_1,r(y|x),\Omega_2)$, the channel capacity is defined by
     \begin{eqnarray*}
     C:=\sup_{q(x)\in  S_1} I(q(x)\cdot r(y|x)).
     \end{eqnarray*}
     
     \end{defi}

Arimoto algorithm is to update from $q^{(t)}(x) \in \mathcal{S}_1$ to
  
     \begin{eqnarray}\label{eq.Arimoto update}
     q^{(t+1)}(x):=\frac{q^{(t)}(x) \exp \{D(r(y|x)||r^{(t)}(y))\}}{\sum_{x'} q^{(t)}(x') \exp \{D(r(y|x')||r^{(t)}(y))\}},
     \end{eqnarray}
     where $r^{(t)}(y)$ means the marginal distribution of $q^{(t)}(x)\cdot r(y|x)$.
     It is known that, by using this algorithm,  $I(q^{(t)}(x)\cdot r(y|x))$ monotonically increases and converges to the channel capacity \cite[Theorem 2]{Arimoto}.
    
\section{\bf Information geometric view of channel capacity in $\mathcal{S}_3$}\label{Inf.geom in 3} 
      Let us try to characterize the channel capacity from the information geometric point of view. In \cite{Ikeda} and \cite{Arimoto}, the the channel capacity in $\mathcal{S}_2$ is referred to. Let us review their outline. A probability distribution that attains the channel capacity satisfies the following interesting condition:
     \begin{theo}\cite[Lemma 1]{Arimoto} \cite[p.554--555]{Ikeda} \label{Thm.geom}
     Assume that a probability distribution $\hat{q}(x) \in {\mathcal{S}_1}$ attains the channel capacity $C$. Then $\hat{q}(x)$ satisfies the following condition:
     \begin{eqnarray}
     D(r(y|x)||r_{\hat{q}}(y))=C~~(\forall x \in \Omega_1)\label{Eq.geom1},
     \end{eqnarray}
     where $r_{\hat{q}}(y)$ denotes the marginal distribution of $\hat{q}(x)\cdot r(y|x)$ on $\Omega_2$.
     Conversely, if there exist $\hat{C}\geq 0$ and $\hat{q}\in \mathcal{S}_1$ satisfying
     \begin{eqnarray}
     D(r(y|x)||r_{\hat{q}}(y))=\hat{C}~~(\forall x \in \Omega_1)\label{Eq.geom3},     
     \end{eqnarray}
     then $\hat{C}\geq 0$ and $\hat{q}(x)$ are the channel capacity and a probability distribution that attains the channel capacity, respectively.    
     \end{theo}
The proof is given in Section \ref{Pr.Thm.geom} for convenience' sake. Theorem \ref{Thm.geom} tells us that, from information geometric view of $\mathcal{S}_2$, the channel capacity is a ``circumcenter'' of the polyhedron spanned by $\{r(y|x)\}^{| \Omega_1 |}_{x=1}$.      

\cite{Matz} refers to an information geometric interpretation of Arimoto algorithm in $\mathcal{S}_2$, using the result of Theorem \ref{Thm.geom}:
\begin{quotation}
Given a current guess $p^{(t)}(x)$, we should check the Kullbuck-Leibler divergences $D( r(y|j) || r_{q^{(t)}}(y) )$ and move the output distribution closer to those $r(y|x)$ for which $D( r(y|j) || r_{q^{(t)}}(y) )$ is large. This can be achived by increasing the respective weights $p^{(t)} (j) $, consistent with the recursion (\ref{eq.Arimoto update}) that increases (decreases) those input probabilities for which $\exp \{ D( r(y|j) || r_{q^{(t)}}(y) ) \}$ is above (below) the average $\sum_{x} q^{(t)}(x) \exp \{D(r(y|x)||r^{(t)}(y))\}$.
\end{quotation}
Although the explanation seems to be valid intuitively, it does not seems to succeed in  revealing the behavior of $r_{q^{(t)}} (y)$ in $\mathcal{S}_2$ as $t$ is updated accurately. Therefore, in our opinion, further researches are needed to reveal the information geometric view of Arimoto algorithm.
      
     In this section, we reconsider the information geometric view of the channel capacity in $\mathcal{S}_3$. We may be able to see some interesting structure in $\mathcal{S}_3$ which is hidden in $\mathcal{S}_2$.
      
      \ Define subsets $M$ and $E$ of $\mathcal{S}_3$ by
      \begin{eqnarray*}
      M&:=&\{q(x)\cdot r(y|x)~ |~q(x) \in \mathcal{S}_1\},
      \\E&:=&\{q(x)\cdot r(y)~ |~q(x) \in \mathcal{S}_1, r(y) \in \mathcal{S}_2\}.
      \end{eqnarray*}
      From Theorem \ref{Theo.sufficient condition to autoparallel}, we can see that  $M$ is $\nabla^{(m)}$-autoparallel and $E$ is $\nabla^{(e)}$-autoparallel.  
      
      \begin{lemm} \label{Lem.proj.to E}
      For $p (x,y) \in \mathcal{S}_3$, the $\nabla^{(m)}$-projection of $p$ onto $E$ is $q(x)\cdot r(y)$,
      where $q(x)$ and $r(y)$ are defined by
      \begin{eqnarray*}
      q(x):=\sum_{y\in \Omega_2} p(x,y),
      \\ r(y):=\sum_{x\in \Omega_1} p(x,y),
      \end{eqnarray*}
      that is, $q(x)$ and $r(y)$ are the marginal distributions of $p(x,y)$.
      \end{lemm} 
      The proof is given in Section \ref{Pr.Lem.proj.to E}. By utilizing Lemma \ref{Lem.proj.to E}, the channel capacity $C$ is expressed as follows:
      \begin{eqnarray}
      C=\sup_{p(x,y) \in M} D(p(x,y)||\Pi^{(m)}(p(x,y)))\label{geom. channnel} ,
      \end{eqnarray}
      where $\Pi^{(m)}(p(x,y))$ means the $\nabla^{(m)}$-projection of $p(x,y)$ onto $E$. The formula (\ref{geom. channnel}) says that, from the viewpoint of geometry in $\mathcal{S}_3$, the channel capacity C is the longest ``distance'' (between $p(x,y)$ and $\Pi^{(m)} (p(x,y))$) from $M$ to $E$ (Fig. \ref{Fig.Comparison}).
     
      \section{\bf Backward em-algorithm}\label{Backward em-algorithm}
    
     In Section \ref{Inf.geom in 3}, we reveal an information geometric structure of the channel capacity in $\mathcal{S}_3$. Therefore, if we can make an algorithm monotonically increasing the Kullback-Leibler divergence, we can expect that this algorithm is useful for evaluating the channel capacity.

     An algorithm, which monotonically decreases the Kullback-Leibler divergence, is well known as ``the em-algorithm'' \cite{S.-I. Amari}. Then how can we increase the Kullback-Leibler divergence? It will be a strong candidate to project onto a $\nabla^{(m)}$($\nabla^{(e)}$)-autoparallel submanifold  by a $\nabla^{(m)}$($\nabla^{(e)}$)-geodesic. But since this projection is a critical point of the Kullback-Leibler divergence, this may sometimes decrease the Kullback-Leibler divergence. Hence, an algorithm that uses this idea is not necessarily a steady algorithm that increases the Kullback-Leibler divergence and converges to the channel capacity $C$.
     
     To overcome this difficulty, let us try to use the idea that rewinds the em-algorithm, same as rewinding movie films!
     \begin{defi}
    Define $\mathcal{S}_3$, $M$ and $E$ in the same way as Section 3. For $q^{(t)}(x)\cdot r(y|x)=:p^{(t)}(x,y) \in M$, update  $q^{(t+1)}(x)\cdot r(y|x)=:p^{(t+1)}(x,y) \in M$ as follows:
     \begin{description}
     \item[\rm 1. Backward e-step.] Search $q_{(t+1)}(x)\cdot r_{(t+1)}(y) \in E$ such that the unique $\nabla^{(e)}$-projection from $q_{(t+1)}(x)\cdot r_{(t+1)}(y)$ onto M is $p^{(t)}(x,y)$. 
     \item[\rm 2. Backward m-step.] Search $q^{(t+1)}(x)\cdot r(y|x) \in M$ such that the unique $\nabla^{(m)}$-projection from $q^{(t+1)}(x)\cdot r(y|x)$ onto E is $q_{(t+1)}(x)\cdot r_{(t+1)}(y)$.
     \end{description}
      We call this algorithm ``the Backward em-algorithm'' (See Fig.\ref{Fig.Comparison}).
     \end{defi}

\begin{figure} 
	\begin{centering}
	\includegraphics[scale=0.6]{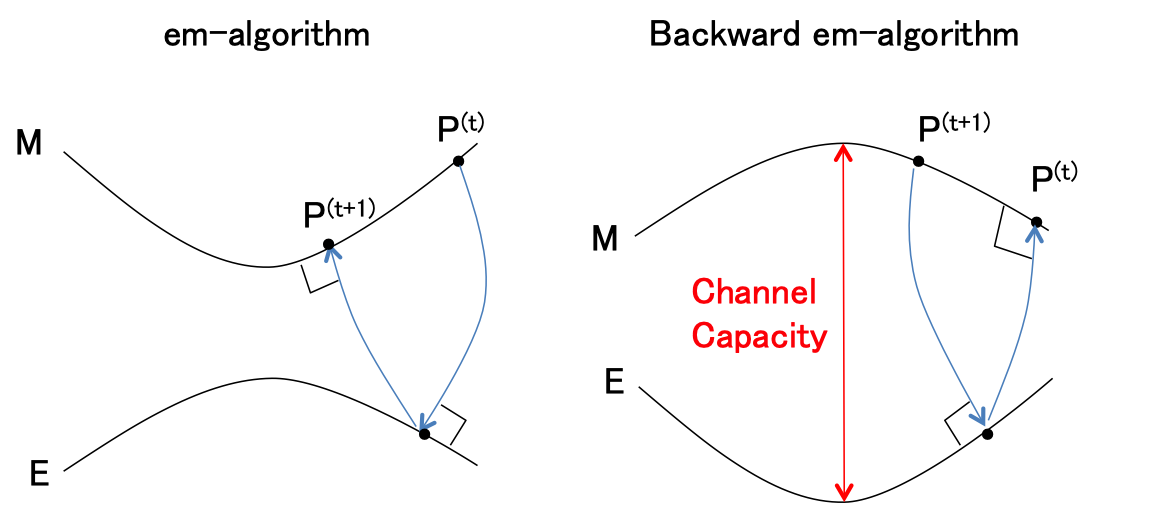}
	\par
	\caption{Comparison between em-algorithm and Backward em-algorithm}\label{Fig.Comparison}
	\end{centering}
	
\end{figure}

     \begin{theo}
     By using the Backward em-algorithm, $I(p^{(t)}(x,y))$ increases as $p^{(t)}(x,y)$ is updated. Namely, the following equality 
     \begin{eqnarray*}
     I(p^{(t)}(x,y))\leq I(p^{(t+1)}(x,y))
          \end{eqnarray*}
     holds.

     \end{theo}
     \begin{proof}
     \begin{eqnarray*}
     I(p^{(t)}(x,y))&=&D(p^{(t)}(x,y)||\Pi^{(m)}({p^{(t)}}(x,y))) \\
                            &\leq& D(p^{(t)}(x,y)||(\Pi^{(m)}{p^{(t)}}(x,y))) +D(\Pi^{(m)}({p^{(t)}}(x,y))||q_{(t+1)}(x)\cdot r_{(t+1)}(y))\\
                            &=& D(p^{(t)}(x,y)||q_{(t+1)}(x)\cdot r_{(t+1)}(y))\\
                            &\leq& D(p^{(t+1)}(x,y)||p^{(t)}(x,y)) +D(p^{(t)}(x,y)||q_{(t+1)}(x)\cdot r_{(t+1)}(y))\\
                            &=& D(p^{(t+1)}(x,y)||q_{(t+1)}(x)\cdot r_{(t+1)}(x))\\
                            &=& I(p^{(t+1)}(x,y)) .                      
      \end{eqnarray*}            
      Note that the second and third equalities follow from the generalized Pythagorean theorem.     
     \end{proof}

     Although we define the Backward em-algorithm, we can not determine whether or not there exist any probability distributions $q_{(t+1)} (x) \cdot r_{(t+1)}(y)$ which satisfy $\Pi^{(e)}( q_{(t+1)}(x) \cdot r_{(t+1)}(y)  ) = p^{(t)}(x,y)$ for a given probability distribution $p^{(t)}(x,y) \in M$. Therefore it is not trivial that we can carry out the Backward e-step. For $p^{(t)}(x,y) \in M$, do there exist any probability distributions $q_{(t+1)}(x) \cdot r_{(t+1)}(y) \in E$ which satisfy $\Pi^{(e)}( q_{(t+1)}(x) \cdot r_{(t+1)}(y)  ) = p^{(t)}(x,y)$ ? And if any, can we write $q_{(t+1)}(x) \cdot r_{(t+1)}(y)$ explicitly? The following theorem answers positively to the above two questions.
       
       \begin{theo} \label{Thm.Backward e-step}
     Let $q^{(t)}(x)\cdot r(y|x) \in M$. Then the following two statements for $q(x) \in \mathcal{S}_1$ and $r(y) \in \mathcal{S}_2$ are equivalent:
     \begin{description}
     \item{1}. $q(x)\cdot r(y) \in E$ satisfies 
     \begin{eqnarray}
     \Pi^{(e)}(q(x)\cdot  r(y)) =q^{(t)}(x)\cdot r(y|x)\label{Statement1}
     \end{eqnarray}
     where $\Pi^{(e)}(q(x)\cdot  r(y))$ denotes the $\nabla^{(e)}$-projection from $q(x)\cdot  r(y)$ onto $M$.
     \item{2}.
     $q(x)\cdot  r(y) \in E$ satisfies
     \begin{eqnarray}
     q(x) \propto q^{(t)}(x)\exp {D(r(y|x)||r(y))}\label{Statement2}.
     \end{eqnarray}
     \end{description}
 \end{theo}

      \begin{proof}
      Fix $q(x)\cdot r(y)$ contained in $E$. Define $L: {\mathbb{R}^{n}_{++}}\times \mathbb{R} \rightarrow \mathbb{R}$ by 
      \begin{eqnarray*}
                  L(\hat{q}(1), ..., \hat{q}(m), \lambda):= D(\hat{q}(x)\cdot r(y|x)||q(x)\cdot r(y))+\lambda (1-\sum_{x=1}^{m} \hat{q}(x)).
      \end{eqnarray*}
      Noting that 
      \begin{eqnarray*}   
      (\ref{Statement1})  \Leftrightarrow \argmin_{\hat{q}(x)\cdot r(y|x) \in M} D(\hat{q}(x)\cdot r(y|x)||q(x)\cdot r(y))=q^{(t)}(x)\cdot r(y|x),
      \end{eqnarray*}
      we see that (\ref{Statement1}) is equivalent to the following:
      \begin{eqnarray*}
\exists \lambda ' ~~s.t. ~~\left.\frac{\partial L}{\partial \hat{q}(i)}\right|_{\hat{q}=q^{(t)},\lambda=\lambda'}=0,~~\left.\frac{\partial L}{\partial \lambda}\right|_{\hat{q}=q^{(t)},\lambda=\lambda'}=0.
      \end{eqnarray*}
      Observing that
      \begin{eqnarray*}
      D(\hat{q}(x)\cdot r(y|x)||q(x)\cdot r(y))
      =D(\hat{q}(x)||q(x))+\sum_{x=1}^{m} \hat{q}(x)D(r(y|x)||r(y)),
      \end{eqnarray*}
      we can see that   
      \begin{eqnarray*}
      \frac{\partial L}{\partial \hat{q}(i)}
      =\log {\hat{q}(i)}-\log{q(i)}+D(r(y|i)||r(y))+1-\lambda,
\end{eqnarray*}
which concludes the proof.
\end{proof}

From Theorem \ref{Thm.Backward e-step}, we can deduce the following interesting  theorem.
      \begin{theo} \label{Theo. e-step is exp.}
      The subset $E^{(t)}$ of $E$ defined by
      \begin{eqnarray*}
     E^{(t)}:=\{q(x)\cdot r(y)~| ~\Pi^{(e)}(q(x)\cdot r(y))=q^{(t)}(x)\cdot r(y|x)~\}
     \end{eqnarray*}
 is $\nabla^{(e)}$-autoparallel.
      \end{theo}
      \begin{proof}
      It suffices to prove that, for any $q_1(x)\cdot r_1(y)$ and $q_2(x)\cdot r_2(y)$ contained in $E^{(t)}$ and any $t$ with $0\leq t\leq 1$, there exists $q_3(x)\cdot r_3(y)$ contained in $E^{(t)}$ satisfying
      \begin{eqnarray}
      t\log (q_1(x)\cdot r_1(y)) +(1-t)\log (q_2(x)\cdot r_2(y))=\log (q_3(x)\cdot r_3(y)) - \Phi_{1} (t) \label{Eq.Thm.Backward e-step1},
      \end{eqnarray}
      where the normalization term $\Phi_{1} (t)$ is defined by
      \begin{eqnarray*}
     \Phi_{1} (t) := \log \left\{ \left\{ \sum_{(x,y) } \exp \left\{ t \log(q_1 (x) \cdot r_{1} (y)) + (1-t) \log(q_2 (x) \cdot r_2 (y) )\right\} \right\}^{-1} \right\}.
      \end{eqnarray*}
      
      Calculating the left-hand side (LHS) of (\ref{Eq.Thm.Backward e-step1}), we obtain
      \begin{center}
       (LHS) $= t\log q_1(x) +(1-t)\log q_2(x)+t\log r_1 (y) +(1-t)\log r_2 (y)$.
      \end{center}
      Let us calculate $t\log q_1 (x) +(1-t)\log q_2(x)$. Noting that the pairs ($q_1(x), r_1(y)$) and ($q_2(x), r_2(y)$) satisfy (\ref{Statement2}), 
      \begin{equation*}
      \begin{split}
    t\log q_1(x) +(1-t)\log q_2(x) &= t\{ \log q^{(t)}(x) +\log \exp D (r(y||x)||r_1(y))+  \Phi_2(t) \}\\   
      &\quad + (1-t)\{ \log q^{(t)} (x) +\log \exp D(r(y||x)||r_2(y))+ \Phi_3(t) \}\\ 
     &= \log q^{(t)}(x)+t(D(r(y||x)||r_1(y)))\\
     &\quad +(1-t) (D(r(y||x)||r_2(y)))+t \Phi_2(t)+ (1-t) \Phi_3(t),
     \end{split}
     \end{equation*}
     where the normalization factors $\Phi_2 (t)$ and $\Phi_3 (t)$ are defined by 
     \begin{eqnarray*}
     \Phi_2(t):= \sum_{x} q^{(t)}(x)\exp D(r(y|x)||r_1(y)),\\
     \Phi_3(t):= \sum_{x} q^{(t)}(x)\exp D(r(y|x)||r_2(y)).
     \end{eqnarray*}   
      Define $r_3 (y)\in \mathcal{S}_2$ by
      \begin{eqnarray}
      t\log r_1 (y) +(1-t)\log r_2 (y) + \Phi_{4} (t)=:\log r_3 (y),\label{Eq.Thm.Backward e-step2}
      \end{eqnarray}
      where 
      \begin{equation*}
      \Phi_{4} (t) := \log \left\{ \left\{ \sum_{y } \exp \left\{ t \log r_{1} (y) + (1-t) \log r_2(y) \right\} \right\}^{-1} \right\}.
      \end{equation*}
      Then we can see that $t(D(r(y||x)||r_1(y)))+(1-t) (D(r(y||x)||r_2(y)))$ can be rewritten as follows by using $r_3(y)$ defined by (\ref{Eq.Thm.Backward e-step2}):
      \begin{equation*}
      \begin{split}
       \sum_{y\in \Omega_2} \{t\cdot r(y||x)\log \frac{r(y|x)}{r_1(y)}&+(1-t)\cdot r(y|x)\log \frac{r(y|x)}{r_2(y)}\}\\
      &=\sum_{y\in \Omega_2} \{r(y|x)\cdot \log r(y|x)-(t\log r_1(y) +(1-t)\log r_2(y))\}\\
       &=\sum_{y\in \Omega_2} \{r(y|x)\cdot \log r(y|x)-\log r_3(y) +  \Phi_{4} (t) \}\\
       &=D(r(y|x)||r_3 (y)) +  \Phi_{4} (t).
      \end{split}
     \end{equation*}
      Set $q_3(x) \in \mathcal{S}_1$ by
      \begin{eqnarray*}
      q_3 (x) := q^{(t)}(x)\exp {D(r(y|x)||r_3 (y))} \cdot \Phi_{5} (t),
      \end{eqnarray*}
      where $\Phi_{5} (t) := \left\{ \sum_{x \in \mathcal{S}_1} q^{(t)}(x)\exp {D(r(y|x)||r_3 (y))} \right\}^{-1}$.
      Then, we obtain
      \begin{eqnarray*}
       t\log q_1  (x)+(1-t)\log q_2(x) =\log q_3 (x) + \Phi_4 (t) + t \Phi_2 (t) + (1-t) \Phi_3 (t) - \log \Phi_5 (t).
      \end{eqnarray*}
     Hence
     \begin{equation}
     \mbox{(LHS)}=\log (q_3 (x) \cdot r_3 (y))+ t \Phi_2 (t) + (1-t) \Phi_3 (t) - \log \Phi_5 (t)\label{Eq.Thm.Backward e-step3}
     \end{equation}
     holds, and therefore it concludes the proof.
      \end{proof}
      
      Theorem \ref{Thm.Backward e-step} and Theorem \ref{Theo. e-step is exp.} tell us that, for any probability distribution $p^{(t)} (x,y) \in M $, we can carry out the Backward e-step and the set $E^{(t)}$ of candidates $q_{(t+1)}(x) \cdot r_{(t+1)} (y) $ for the Backward e-step is an exponential family.
            
      Next, let us consider whether or not we can carry out the Backward m-step. Which element should we choose in $E^{(t)}$ to carry out the Backward m-step? That is, what are conditions of $q_{(t+1)}(x) \cdot r_{(t+1)} (y) \in E^{(t)}$ that there exists $p^{(t+1)}(x,y) \in M$ such that $\Pi^{(m)} (p^{(t+1)}(x,y) ) = q_{(t+1)}(x) \cdot r_{(t+1)} (y)$ holds?

      To investigate this question, let $\Pi^{(m)} (M) $ be the embedding  of $M$ into $E$ by $\nabla ^{(m)}$-projection. Assume that there exist any intersections of $\Pi^{(m)} (M)$ with $E^{(t)}$ (its existence and uniqueness is discussed in Section \ref{section.Concluding Remarks}). Let $\hat{q}_{(t+1)} (x) \cdot \hat{r}_{(t+1)} (y) \in \Pi^{(m)} (M) \bigcap E^{(t)}$. Then, for $\hat{q}_{(t+1)} (x) \cdot \hat{r}_{(t+1)} (y)$, we can carry out the Backward m-step. Conversely, assume that, for $\hat{q}_{(t+1)} (x) \cdot  \hat{r}_{(t+1)} (y) $, we can carry out the Backward m-step. Then, $\hat{q}_{(t+1)} (x) \cdot \hat{r}_{(t+1)} (y) \in \Pi^{(m)} (M) \bigcap E^{(t)}$.  Hence, the problem of searching $q_{(t+1)} (x) \cdot r_{(t+1)} (y) \in E^{(t)}$ where the Backward m-step can be carried out is equivalent to the one  of searching any intersections of $\Pi^{(m)} (M)$ with $E^{(t)} $.
      
      The element $q_{(t+1)} (x) \cdot r_{(t+1)} (y) \in E^{(t)}$ is only depend on $r_{(t+1)} (y)$ because, for a  given $r_{(t+1)} (y)$, the requirement that $q_{(t+1)} (x) \cdot  r_{(t+1)} (y)  \in E^{(t)} $ determine $q_{(t+1)} (x)$ by the equation (\ref{Statement2}).  Therefore, from now on, we may see $q_{(t+1)} (x)$ as as a function of $r_{(t+1)} (y)$ determined by the requirement that $q_{(t+1)} (x) \cdot r_{(t+1)} (y) \in E^{(t)}$ ($i.e.$, the equation (\ref{Statement2})). Taking it into consideration, we may consider the condition of $r_{(t+1)} (y)$ such that 
     \begin{eqnarray}
     \exists q^{(t+1)}(x) \in \mathcal{S}_1 ~s.t.~ \Pi^{(m)} (q^{(t+1)} (x) \cdot r(y|x) ) = q_{(t+1)} (x) \cdot r_{(t+1)} (y) \label{requirement of Backward m-step}
     \end{eqnarray}
     holds.
     Noting that $\Pi^{(m)} (q^{(t+1)} (x) \cdot r(y|x) ) = q^{(t+1)} (x) \cdot r_{q^{(t+1)} }(y)$ (See Theorem \ref{Lem.proj.to E}), where $r_{q^{(t+1)}}(y)$ denotes the marginal distribution of $q_{(t+1)}(x)\cdot r(y|x)$, the condition (\ref{requirement of Backward m-step}) of  $r_{(t+1)} (y)$ is equivalent to the following condition:
     \begin{equation}
\exists  q_{(t+1)} (x)  \in \mathcal{S}_1 ~ s.t.  ~ \left \{
\begin{array}{l}
q^{(t+1)} (x) = q_{(t+1)} (x)\\
r_{(t+1)} (y) = r_{q^{(t+1)}} (y).
\end{array}
\right.
\end{equation}
     The above condition comes down to solving the following nonlinear equation with respect to  $r_{(t+1)}(y)$:
     \begin{eqnarray} \label{Eq.Backward m-step}
     r_{(t+1)} (y) = r_{q_{(t+1)}}  (y) .
     \end{eqnarray}

            Rewritting this as 
   \begin{eqnarray*}
      \Phi(t,r_{(t+1)})^{-1}\sum_{x\in \Omega_1} \left\{ q^{(t)}(x) \cdot \left\{ \exp \sum_{y\in \Omega_2}r(y|x)\log \frac{r(y|x)}{r_{(t+1)} (y)}\right\} \cdot r(y|x)\right\}=r_{(t+1)} (y)
\end{eqnarray*}
      \begin{eqnarray*}
       \left(\Phi(t,r_{(t+1)} ):=\sum_{x\in \Omega_1} q^{(t)}(x)\exp D\left(r(y|x)||r_{(t+1)} (y)\right) \right),
\end{eqnarray*}
we see that it is difficult to solve the nonlinear equation (\ref{Eq.Backward m-step}) with respect to $r_{(t+1)} (y)$. If we can solve the equation (\ref{Eq.Backward m-step}), we can prove that $I(p^{(t)}(x,y))$ converges to 
the channel capacity $C$. The proof is given in Section \ref{Convergence}.
      
As it is difficult to solve the equation (\ref{Eq.Backward m-step}) with respect to $r_{(t+1)}(y)$, we try to approximate (\ref{Eq.Backward m-step}) in order that we can solve. It will be a good solution to approximate      
      $\exp (D(r(y|x)||r_{(t+1)} (y))$ to some value that is independent of $x$ since it becomes a constant value. It seems good to approximate $r(y|x)$ of $\exp (D(r(y|x)||r_{(t+1)} (y))$ to the ``circumcenter'' $r^{*} (y) $ of the figure induced from $\{r(y|x)\}_{x=1}^{| \Omega_1 |}$ in $\mathcal{S}_2$, that is, the probability distribution contained in $\mathcal{S}_2$ that attains the channel capacity (see Theorem \ref{Thm.geom}). Then, observing that $\exp (D(r(y|x)||r_{(t+1)} (y))$ becomes independent of $x$, (14) is rewtitten as
      
      \begin{eqnarray}
      \sum_{x\in \Omega_1} (q^{(t)}(x)\cdot r(y|x))=r_{(t+1)}(y)\label{Eq.Revised.Backward m-step},
      \end{eqnarray}      
      that can be solved.  The merit of this approximation is that $ r^{*}(y) $ has dissapearred in the equation (\ref{Eq.Revised.Backward m-step}). Namely, even if we do not know the value of $r^{*}(y)$, we can solve (\ref{Eq.Revised.Backward m-step}). Since the solution of (\ref{Eq.Revised.Backward m-step}) is $r_{(t+1)} (y)=r_{q^{(t)}}(y)$, the approximation designates the element of $E^{(t)}$ by $q_{(t+1)}(x) \cdot r_{q^{(t)}} (y)$ where $q_{(t+1)}(x) \in \mathcal{S}_1$ is defined by
      \begin{eqnarray*}
      q_{(t+1)}(x)\propto q^{(t)}(x)\exp {D(r(y|x)||r_{q^{(t)}}(y))}.
      \end{eqnarray*}
      In the present paper, we call the approximation the approximate Backward e-step.
      
      By the above approximation, we can solve the equation (\ref{Eq.Backward m-step}). But, in return for the approximation, $q_{(t+1)}(x) \cdot r_{q^{(t)}} (y)$ is not necessarily an intersection of $\Pi^{(m)} (M) $ with $E^{(t)}$ , and therefore we need to approximate the Backward m-step too. In the present paper, we approximate the Backward m-step by the $\nabla^{(m)}$-projection of $q_{(t+1)}(x) \cdot r_{q^{(t)}} (y)$ onto $M$. A short computation shows that 
\begin{eqnarray}  \label{eq.Lem.proj. to M}   
      \Pi^{(m)} \left( q_{(t+1)}(x) \cdot r_{q^{(t)}} (y) \right)=  q_{(t+1)}(x) \cdot r(y |x). 
\end{eqnarray}    
      The proof is given in Section \ref{Pr.Lem.proj. to M}.  We call this approximation the approximate Backward m-step.
      
      Combining the approximate Backward m- and e-steps, $q^{(t)} (x) \cdot r(y|x)$ is updated by $q_{(t+1)} (x)\cdot r(y|x) $, and therefore, $q^{(t)}(x) $ is updated by $q_{(t+1)} (x) $, which is nothing but Arimoto algorithm (Fig. \ref{fig:Arimoto}).   
      
 \begin{figure}

\centering       
       
 \begin{tikzpicture}

\path[name path=border1] (0,0) to[out=-10,in=150] (6,-2);
\path[name path=border2] (12,1) to[out=150,in=-10] (5.5,3.2);
\path[name path=redline] (0,-0.4) -- (12,1.5);
\path[name intersections={of=border1 and redline,by={a}}];
\path[name intersections={of=border2 and redline,by={b}}];
\shade[left color=gray!10,right color=gray!80] 
  (0,0) to[out=-10,in=150] (6,-2) -- (12,1) to[out=150,in=-10] (5.5,3.7) -- cycle;

\draw (-0.3,3.5) to[out=40,in=170] 
  coordinate[pos=0.27] (aux1) 
  coordinate[pos=0.6] (aux2)
  coordinate[pos=0.9] (aux3) 
  coordinate[pos=0.94] (aux4)  
   coordinate[pos=0.96] (aux5)  
   coordinate[pos=0.7] (aux6)
   coordinate[pos=0.75] (aux7)(10,7.5);
\draw[name path=EE^(t), dashed] (a) .. controls (6,1.5) and (7,2) .. 
  coordinate[pos=0.2] (bux1)  
  coordinate[pos=0.695] (bux2)   (b);

\draw[name path=E^(t),thick] (10.4,0.2) .. controls (8,3) and (4,3) .. 
coordinate[pos=0.1] (eux1) (4.2,2.8);
\path[name intersections={of=E^(t) and EE^(t),by={c}}];

\foreach \coor in {1}
  \draw[dashed,->] (aux\coor)to[bend left] (3.5,3);

\draw[ultra thick,dotted, red, <-] (aux3)to[bend left] 
coordinate[pos=0.07] (cux1)     (c);
\node[red, font=\sffamily,rotate=-35] at (c) {X};

\draw[ultra thick,dotted, red, ->] (aux2)to[bend left] 
coordinate[pos=0.92] (dux1)     (c);

\draw[thick, red] (aux4)--(9.4,7.2);
\draw[thick, red] (9.4,7.2)--(cux1);

\draw[thick, red] (bux2)--(8,2.2);
\draw[thick, red] (8,2.2)--(dux1);

\draw[ultra thick,loosely dashed, blue, ->] (eux1)[bend left].. controls (11,3) and (11,6) .. 
coordinate[pos=0.87] (fux1) (aux3);
\draw[fill=blue] (eux1) circle (3pt);

\draw[ultra thick,loosely dashed, blue, ->] (eux1)[bend left].. controls (10,3) and (11,5) .. 
coordinate[pos=0.95] (gux1) (aux6);

\draw[thick, blue] (aux5)--(10,7);
\draw[thick, blue] (10,7)--(fux1);

\draw[thick, blue] (aux7)--(7.9,7.2);
\draw[thick, blue] (7.9,7.2)--(gux1);

\node[rotate=30] at (6.2,-1.5) {$E$};
\node[rotate=30] at (10.3,7.7) {$M$};
\node[rotate=20] at (2.5,4.5) {${\Pi}^{(m)}$};
\node[rotate=10] at (2.9,0) {${\Pi}^{(m)}(M)$};
\node[rotate=20] at (5,2.5) {$E^{(t)}$};

\node at ([yshift=-1cm]current bounding box.south) {
\setlength\tabcolsep{3pt}
\begin{tabular}{@{}cl@{\hspace{12pt}}cl@{}}
 
\tikz\draw[ultra thick,dotted,red] (0,0.5) --(1.5,0.5); & Backward em-algorithm
  & \tikz\draw[ultra thick,loosely dashed,blue] (0,1) --(1.5,1); & Arimoto algorithm
 \\ 
 
 \textcolor{red}{\sffamily  X} & An exact solution of the equation (\ref{Eq.Backward m-step}) & \tikz\draw[fill=blue] (0,0) circle (2pt); & An approximate solution of (\ref{Eq.Backward m-step})
\end{tabular}
};

\end{tikzpicture}

\caption{Information geometric view of Arimoto algorithm.}\label{fig:Arimoto}

\end{figure}
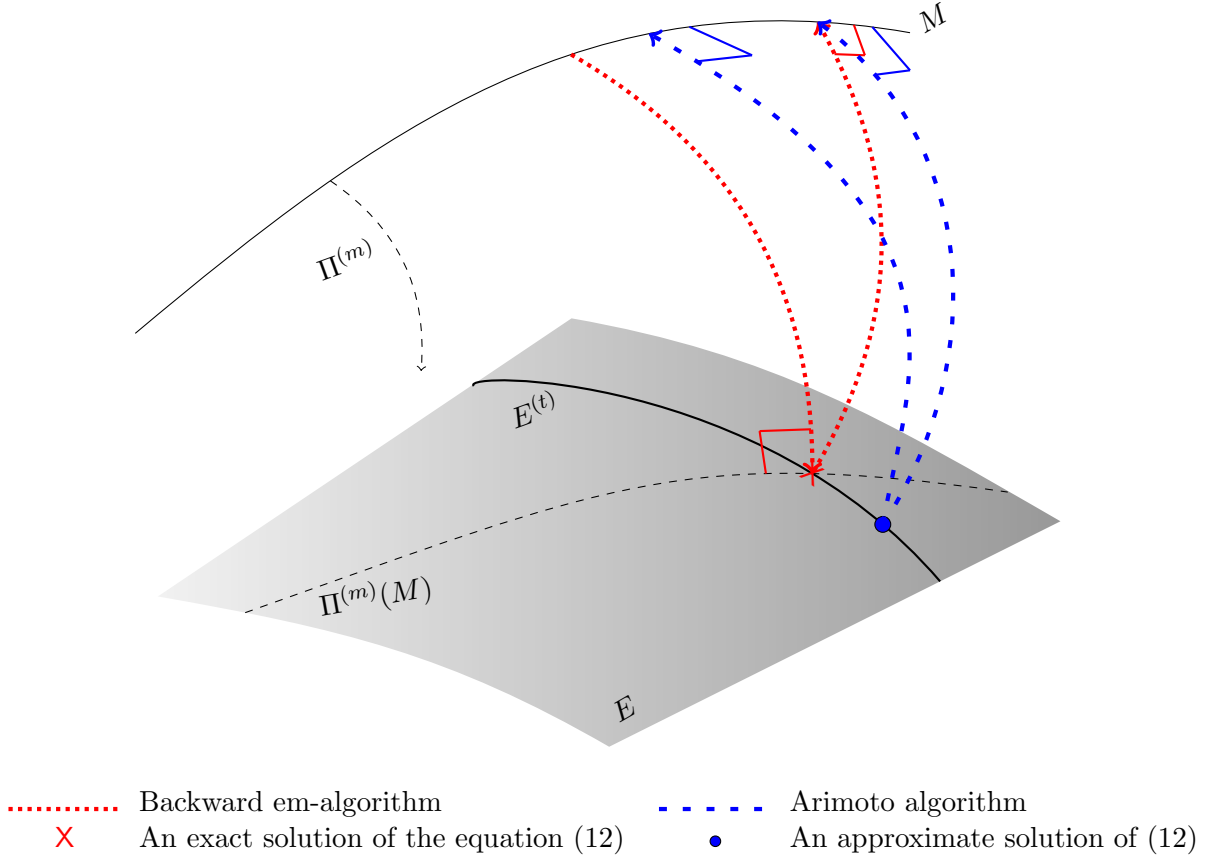

     \section{\bf Concluding Remarks}\label{section.Concluding Remarks}
In the present paper, we investigated the channel capacity from the information geometric point of view in $\mathcal{S}_3$.  Then, we introduced the new algorithm that monotonically increases the Kullback-Leibler divergence, ``the Backward em-algorithm.'' The Backward e-step can be determined but the Backward m-step cannot be. Hence, we tried to approximate the Backward m-step, which corresponds to Arimoto algorithm.

There are many open problems left. First, existence and uniqueness of an intersection of $\Pi{(M)}$ with $E^{(t)}$ should be studied. To research the problem, we may consider the uniqueness and existence of a solution in the equation (\ref{Eq.Backward m-step}). If we can prove that there exists a solution of the equation (\ref{Eq.Backward m-step}), even if we cannot solve, we may be able to introduce other approximations of the Backward e- and m-steps and accelerate Arimoto algorithm.
 
It seems interesting to apply the Backward em-algorithm to other subjects. In our knowledge, there has been no algorithm that monotonically increases the Kullback-Leibler divergence. We can use the Backward em-algorithm when we want to increase the Kullback-Leibler divergence between two manifolds. For example, in the field of independent component analysis and machine learning, we often need to increase the mutual information (e.g., \cite{Lee,Petar}). In these situations, there is a possibility that the Backward em-algorithm  works well because information geometric view of the mutual information is the Kullback-Leibler divergence between two manifolds.

\section*{Acknowledgments}
The author gratefully acknowledges the continuous encouragement from Toru Ohira and Hideyuki Ishi.
The author also thanks Professors Amor Keziou, Hiroshi Matsuzoe, Masahito Hayashi, Shiro Ikeda and Phillippe Regnault  for their helpful discussions and comments.
 
\section{Appendices} 

\subsection{ Proof of Theorem \ref{Thm.geom}} \label{Pr.Thm.geom}

(First half): Define a function $L: {\mathbb{R}^{m}_{++}}\times \mathbb{R} \rightarrow \mathbb{R}$ by
\begin{eqnarray*}
     L(q_1,...,q_m,\lambda):=\sum_{i} q_i D(r(y|i)||r_{q}(y))+\lambda (\sum_i {q_i} -1),
\end{eqnarray*}
where $\lambda$ means a Lagrange multiplier. For the mutual information to take a maximum point at $\hat{q} \in {\mathcal{S}_1}$, it is necessary that 
\begin{eqnarray}
    \exists \hat{\lambda} ~~s.t. ~~\left.\frac{\partial L}{\partial \lambda}\right|_{\lambda=\hat{\lambda},q=\hat{q}}=0,~~~\left.\frac{\partial L}{\partial q_i}\right|_{\lambda=\hat{\lambda},q=\hat{q}}=0 \label{Eq.Lagrange}.
\end{eqnarray}

Since ${\partial L}/{\partial q_i}=D(r(y|i)||r_q (y))-1+\lambda$, (\ref{Eq.Lagrange}) is rewritten as 
\begin{eqnarray*}
\sum_i \hat{q}_i -1=0,~~D(r(y|i)||r_{\hat{q}} (y))=(1-\hat{\lambda})~(\forall i),
\end{eqnarray*}
and it follows immediately that $1-\hat{\lambda}$ corresponds to the channel capacity $C$ and the relation (\ref{Eq.geom1})  holds.\\
\noindent
(Second half): The minmax redundancy, defined by
\begin{eqnarray*}
\min_{q\in \overline{\mathcal{S}_1}} \max_{x\in \Omega_1} D(r(y|x)||r_q (y)),
\end{eqnarray*}
coincides with the channel capacity \cite[Theorem13.1.1]{Cover}, where $r_q (y)$ is the marginal distribution of $q(x)\cdot r(y|x)$. Since, for $\hat{q}$ $\in \mathcal{S}_1$ satisfying (\ref{Eq.geom3}), the equality $\max_{x\in \Omega_1} D(r(y|x)||r_{\hat{q}} (y))=\hat{C}$ holds, it follows that $C \leq \hat{C} $. Noting that the $\hat{q}$ satisfies $I(\hat{q} (x) \cdot r(y|x))=\hat{C}$ and taking the definition of the channel capacity into the consideration, it also follows that $C \geq \hat{C}$, and therefore, $C=\hat{C}$.

\subsection{  Proof of Lemma \ref{Lem.proj.to E}} \label{Pr.Lem.proj.to E}

Take any $\hat q \cdot \hat r$ contained in $E$. Then 
\begin{eqnarray*}
D(p(x,y)||\hat q (x)\cdot \hat r(y))&-&D(p(x,y) ||q (x)\cdot r(y))
\\ &=&\sum_{x,y}( p(x,y)\log {\frac{p(x,y)}{\hat q (x) \cdot \hat r (y) }}-p(x,y)\log {\frac{p(x,y)}{q (x)  \cdot r (y) }})
\\ &=&\sum_{x,y} \left(-p(x,y)\log \left( {\hat q (x)\cdot \hat r (y) } \right) +p(x,y)\log \left( {q (x)\cdot r(y)} \right) \right)
\\ &=&\sum_{x,y} p(x,y)(\log {\frac{q(x)}{\hat q(x)}}+\log {\frac{r(y)}{\hat r(y)}})
\\ &=&D(q(x)||\hat q(x))+D(r(y)||\hat r(y))
\\ &\geq& 0
\end{eqnarray*}
holds and the lower bound $0$ is attained if and only if $\hat q (x)\cdot \hat{r} (y)=q(x)\cdot r(y)$ (since $D(p_1||p_2) \Leftrightarrow p_1=p_2$ holds \cite[p.31]{Cover}). Observing that the ${\nabla}^{(m)}$-projection $\Pi^{(m)}\left( p(x,y) \right)$ onto the ${\nabla}^{(e)}$-autoparallel submanifold $E$ is characterized by
\begin{eqnarray*}
\Pi^{(m)} \left( p(x,y) \right)=\argmin_{\hat q (x) \cdot \hat r (y) \in E} D(p(x,y) ||\hat q (x) \cdot \hat r (y) ),
\end{eqnarray*}
it concludes the proof. 

\subsection{  Proof of the equation (\ref{eq.Lem.proj. to M}) }      \label{Pr.Lem.proj. to M}

It suffices to prove the following lemma.
\begin{lemm}
Let $q(x) \cdot r(y)  \in E$. Then  $q(x) \cdot r(y|x) $ is one of the candidates for $\nabla^{(m)} $- projections of $q(x) \cdot r(y)$ onto $M$.
\end{lemm}
Note that since $M$ is not $\nabla^{(e)}$-autoparallel, a $\nabla^{(m)} $- projection onto $M$ is not necessarily unique.

\begin{proof}
Take any $\hat q (x) \cdot r(y|x) $ contained in $M$. Then 
\begin{eqnarray*}
D(q(x) \cdot r(y) ||\hat q (x)\cdot \hat r(y|x) )&-&D(q(x) \cdot r(y)  ||q(x) \cdot r(y|x) )
\\&=&\sum_{x,y}q(x)\cdot r(y) \left( \log {\frac{q(x) \cdot r(y)}{\hat q (x) \cdot r(y|x)}}-\log {\frac{q(x) \cdot r(y)}{q (x) \cdot r(y|x)}}) \right)
\\ &=&\sum_{x,y} q(x)\cdot r(y) \left( \log {\frac{q(x) \cdot r(y|x)}{\hat q (x) \cdot r(y|x)}} \right)
\\ &=&D(q(x)||\hat q(x))
\\ &\geq& 0
\end{eqnarray*}
holds and the lower bound $0$ is attained if and only if $\hat q(x) =q(x) $. Observing that, if
\begin{eqnarray*}
p(x, y) =\argmin_{\hat p (x,y) \in M } D(q(x) \cdot r(y) || \hat p (x,y))
\end{eqnarray*}
holds, $p(x,y)$ is one of the candidates for $\nabla^{(m)} $- projections of $q(x) \cdot r(y)$ onto $M$\cite[Theorem 3.10]{S. Amari and H. Nagaoka},
it concludes the proof. 
\end{proof}

\subsection{  Convergence of the Backward em-algorithm} \label{Convergence}

In this section, we assume that the equation (\ref{Eq.Backward m-step}) can be solved and that $p^{(t)}(x,y)$ can be updated to $p^{(t+1)}(x,y)$ any number of times by the Backward em-algorithm. 
\begin{theo} \label{Thm.convergence}
$I(p^{(t)})$ converges to the channel capacity $C$ as $p^{(t)}$ is updated.
\end{theo}

\begin{lemm}
Let $q\cdot r \in E^{(t)}$. Then,

\begin{eqnarray*}
D(p^{(t)}||q\cdot r)=\log \Phi(t,r),
\end{eqnarray*}

where $\Phi(t,r):=\sum_{x\in \Omega_1} q^{(t)}(x) \exp D(r(y|x)||r(y))$.
\end{lemm}
\begin{proof}
      \begin{equation*}
      \begin{split}
      D(p^{(t)}||q\cdot r)&=\sum_{x,y} p^{(t)}(x,y) \log \frac{q^{(t)}(x)\cdot r(y|x)}{q(x)\cdot r(y)} \\
      &=\sum_{x,y} p^{(t)}(x,y)\log \frac{q^{(t)}(x)\cdot r(y|x)}{r(y)q^{(t)}(x) \exp D(r(y|x)||r(y)) \Phi(t,r)^{-1}}\\
      &=\sum_{x,y}q^{(t)}(x)\cdot r(y|x) \{ \log \frac{r(y|x)}{r(y)} -D(r(y|x)||r(y))+\log{\Phi (t,r)} \}\\
      &= \sum_x q^{(t)}(x) \sum_y r(y|x) \log \frac{r(y|x)}{r(y)} -\sum_x q^{(t)}(x)D(r(y|x)||r(y)) +\log \Phi (t,r)\\
      &= \log \Phi (t,r).
      \end{split}  
      \end{equation*}
\end{proof}     

{\bf Proof of Theorem \ref{Thm.convergence}.}
 It suffices to prove that $D(p^{(t)}||p_{(t+1)})$ converges to the channel capacity $C$. Let $q^{(0)}(x) \in \overline{\mathcal{S}_1}$ be a probability distribution that attains the channel capacity $C$. First, let us prove that 
\begin{eqnarray}
C-D(p^{(t)}||p_{(t+1)}) \leq \sum_x q^{(0)}(x) \log \frac{q^{(t+1)}(x)}{q^{(t)}(x)}.\label{Eq.Convergence}
\end{eqnarray}
Calculating $\sum_x q^{(0)}(x) \log \{q^{(t+1)}(x)/q^{(t)}(x)\}$, we obtain
\begin{eqnarray*}
\sum_x q^{(0)}(x) \log \frac{q^{(t+1)}(x)}{q^{(t)}(x)} &=& \sum_x q^{(0)}(x) \log \frac{q^{(t)}(x) \exp D(r(y|x)||r_{(t+1)}(y))\Phi(t,r_{(t+1)})^{-1}}{q^{(t)}(x)}\\
                                                      &=& -\log \Phi(t,r_{(t+1)}) +\sum_x q^{(0)}(x)D(r(y|x)||r_{(t+1)}(y))\\                                                     
                                                      &=& -D(p^{(t)}||p_{(t+1)}) +\sum_{x,y} q^{(0)}(x)r(y|x) \log \frac{r(y|x)}{r_{(t+1)}(y)}\\
                                                      &=& -D(p^{(t)}||p_{(t+1)}) +\sum_{x,y} q^{(0)}(x)r(y|x) \{ \log \frac{r(y|x)}{r_{q^{(0)}}(y)} +\log \frac {r_{q^{(0)}}(y)}{r_{(t+1)}(y)}\}\\
                                                      &=& -D(p^{(t)}||p_{(t+1)}) +\sum_x q^{(0)}(x)D(r(y|x)||r_{q^{(0)}}(y)) +D(r_{q^{(0)}}(y)||r_{(t+1)}(y))\\
                                                      &=& -D(p^{(t)}||p_{(t+1)})+C+D(r_{q^{(0)}}(y)||r_{(t+1)}(y))\\
                                                      &\geq&  -D(p^{(t)}||p_{(t+1)})+C   ,                                                                                                
\end{eqnarray*}
and therefore, we obtain the inequality (\ref{Eq.Convergence}), where $r_{q^{(0)}}(y)$ denotes the marginal distribution of $q^{(0)}(x)\cdot r(y|x)$ on $\Omega_2$.
Summing up the both sides of the inequality (\ref{Eq.Convergence}), we have
\begin{eqnarray*}
\sum_{t=1}^{T} (C-D(p^{(t)}||p_{(t+1)})) &\leq& \sum_x q^{(0)}(x) \log \frac {q^{(T+1)}(x)}{q^{(1)}(x)}\\
                                         &\leq& \sum_x q^{(0)}(x) \log \frac {q^{(0)}(x)}{q^{(1)}(x)}=D(q^{(0)}(x)||q^{(1)}(x)).
\end{eqnarray*}
Noting that $0\leq  D(q^{(0)}(x)||q^{(1)}(x))< \infty$ and is independent of $t$, we can see that the sequence  $\{C-D(p^{(t)}||p_{(t+1)})\}_{t=1}^{\infty}$ converges $0$.
\qed

%
%



\end{document}